\documentclass[review]{elsarticle}
\usepackage{graphicx}
\usepackage{amssymb}
\usepackage{amsmath}
\usepackage{lineno}
\modulolinenumbers[5]

\journal{arXiv}

\newtheorem{theorem}{Theorem}
\newtheorem{lemma}{Lemma}

\newproof{proof}{Proof}

\begin{document}

\begin{frontmatter}

\title{RAC Drawings in Subcubic Area}

\author[mymainaddress]{Zahed Rahmati\corref{mycorrespondingauthor}}
\ead{zrahmati@aut.ac.ir}
\cortext[mycorrespondingauthor]{Corresponding author}

\author[mymainaddress]{Fatemeh Emami}
\ead{fatima.emami13}

\address[mymainaddress]{Department of Mathematics and Computer Science, Amirkabir University of Technology}

\begin{abstract}
In this paper, we study tradeoffs between curve complexity and area of Right Angle Crossing drawings (RAC drawings), which is a challenging theoretical problem in graph drawing. Given a graph with $n$ vertices and $m$ edges, we provide a RAC drawing algorithm with curve complexity $6$ and area $O(n^{2.75})$, which takes time $O(n+m)$. Our algorithm improves the previous upper bound $O(n^3)$, by Di Giacomo~\emph{et al.}, on the area of RAC drawings.
\end{abstract}


\end{frontmatter}


\section{Introduction}

Drawing non-planar graphs with forbidden or desired crossing configurations is a topic that has received growing attention in the last decade; this topic is often recognized as graph drawing beyond planarity~\cite{DidimoSurvey2019}. In addition to guaranteeing specific properties for the edge crossings, some other geometric optimization goals can be considered, such as, for example, minimizing the number of edge bends or minimizing the area occupied by the drawing~\cite{Purchase1997}. The trade-off between maximum number of bends per edge and area requirement has been particularly studied for {\it RAC drawings} (Right Angle Crossing drawings), {\it i.e.}, drawings in which every two crossing edges are orthogonal~\cite{DIDIMO20115156}. RAC drawings are motivated by cognitive studies showing that large crossing angles do not affect too much the readability of a graph layout~\cite{4126225,4475457,HUANG2014452}. 
Deciding whether a graph admits a ($1$-planar) straight-line RAC drawing is NP-hard~\cite{SOFSEM2011Argyriou,BEKOS201748}. 
Didimo~{\it et al.}~\cite{DIDIMO20115156} proved that graphs that admit a straight-line RAC drawing have at most $4n-10$ edges, which is a tight bound. 

A {\it $k$-bend RAC drawing} is a RAC drawing where each edge is a polyline with at most $k$ bends, $k>0$. 
The {\it curve complexity} of a polyline drawing is the maximum number of bends per edge in the drawing. The {\it area} of a drawing is the size of the smallest axis-aligned box enclosing the drawing on an integer grid, in which the vertices and bends are located at grid points.
The $1$-bend RAC drawings with $n$ vertices cannot have more than $5.5n-O(1)$ edges and there are infinitely many $1$-bend RAC graphs  with exactly $5n-O(1)$ edges~\cite{GD201Angelini8}. The $2$-bend RAC drawings can have at most $74.2n$ edges~\cite{ARIKUSHI2012169}.

Didimo~\emph{et al.}~\cite{DIDIMO20115156} proved that every graph has a RAC drawing with curve complexity $3$ (which can be drawn in area $O(n^4)$). 
Di Giacomo~\emph{et al.}~\cite{DiGiacomo2011} studied the tradeoffs between area and curve complexity. They proved that, by increasing the curve complexity to $4$, the area can be improved: Every graph admits a RAC drawing with curve complexity $4$ and area $O(n^3)$. 

Didimo and Liotta~\cite{Pach2013} posed a {\it theoretical open problem}: Is it possible to compute RAC drawings of graphs in $o(n^3)$ area with curve complexity 4? 

Recently, F{\"o}erster and Kaufmann~\cite{GD2019Foerster} presented some contributions as a poster in GD~2019: Every graph on $n$ vertices admits a RAC drawing with 3 bends per edge in $O(n^3)$ area, which improves the result in~\cite{DiGiacomo2011}; there exists no RAC drawing of $K_n$ with 3 bends per edge in $O(n^2)$ area for sufficiently large $n$; every graph admits a RAC drawing with 8 bends per edge in $O(n^2)$ area. Our result, in this paper, is somehow in the middle.

\paragraph{Our contribution}
Given a graph $G$ with $n$ vertices and $m$ edges, we provide a RAC drawing algorithm for drawing $G$ on an integer grid with curve complexity $6$ and area $O(n^{2.75})$. Our algorithm takes time $O(n+m)$.
We in fact obtain improvement on the previous result on the area of RAC drawings by Di Giacomo~\emph{et~al.}~\cite{DiGiacomo2011}, by adding two more bends on each edge. In other words, since the area of a RAC drawing for $K_n$ is $\Omega(n^2)$~\cite{DiGiacomo2011}, we make the gap between the lower and upper bound on the area of RAC drawings smaller.
\section{$6$-Bend RAC drawing in area $O(n^{2.75})$}
In order to prove that every $n$-vertex graph admits a RAC drawing with curve complexity 6 in $O(n^{2.75})$, we describe how to compute such a drawing for the complete graph $K_n$.
\paragraph{\bf Placing the vertices}
Assume $\sqrt[4]{n}$ is an integer. Partition the vertices of $K_n$ into $\sqrt{n}$ groups, each with $\sqrt{n}$ vertices, and place the vertices of each group at a level where the vertical distance of every two consecutive levels is $8\sqrt[4]{n^3}+\sqrt[4]{n}+1$. Denote the $\sqrt{n}$ vertices of level $i$ by $V_{i,1},\dots,V_{i,\sqrt{n}}$, from left to right, respectively. The {\it highest level} is where $i=1$ and the {\it lowest level} is where $i=\sqrt{n}$. 

For each level $i$, we place the vertices $V_{i,j}$, $j=1,\dots,\sqrt{n}$, in such a way that all the vertices have the same vertical coordinate and the distance between two consecutive vertices $V_{i,j}$ and $V_{i,{j+1}}$ is $n+1$. For two consecutive levels $i$ and $i+1$, we place their first vertices $V_{i,1}$ and $V_{i+1,1}$ in such a way that the horizontal distance between $V_{i,1}$ and $V_{i+1,1}$ is $\sqrt{n}+8$, \emph{i.e.,} the vertices of the next level are shifted to the right $\sqrt{n}+8$ units; see Figure~\ref{fig:PlacingPnts}.
\begin{figure}
\centering
\includegraphics[scale=.9]{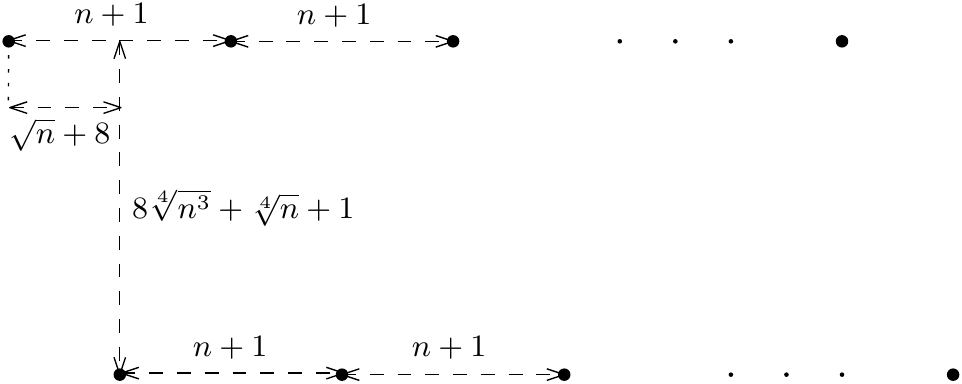}
\caption{Placing the vertices of two consecutive levels.}
\label{fig:PlacingPnts}
\end{figure}

Next we show how to draw the edges of $K_n$, each with six bends. We use two slopes $\alpha={1/ {\sqrt[4]{n^3}}}$ and $-1/ \alpha$ for the {\it intersecting} segments of the edges.
In our technique, we draw the edges from a vertex $V_{i,j}$ at a {\it higher} level and bend the edges to get connected to some other vertex $V_{u,w}$ at a {\it lower} level.
Figure~\ref{fig:RACdrawing} shows a RAC drawing of some portion of $K_{16}$, which in fact would be a RAC drawing for $K_8$, with our technique that we describe as follows.
\begin{figure}
\centering
\includegraphics[scale=.36]{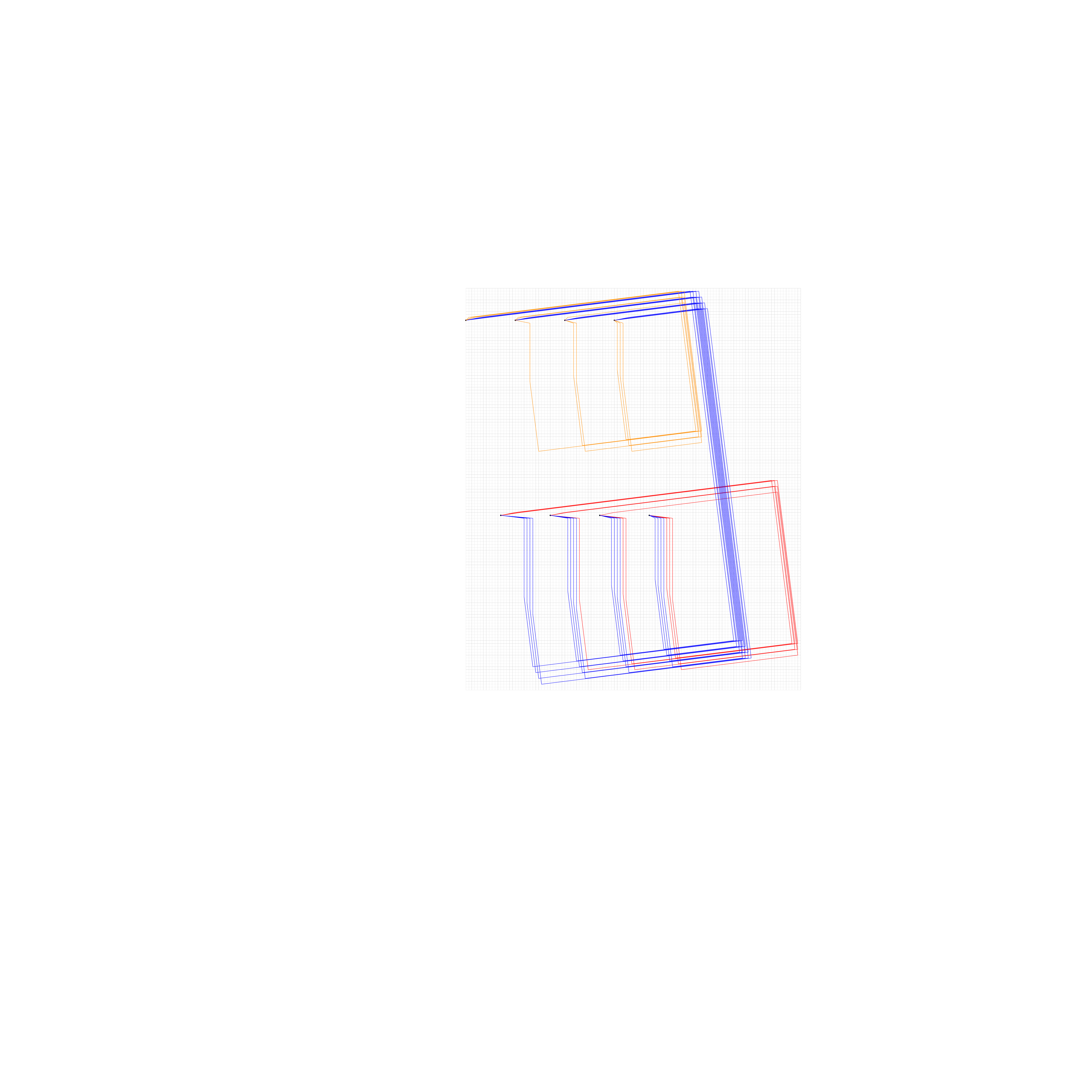}
\caption{A RAC drawing of some portion of $K_{16}$, that is the drawing of the edges incident to vertices of the first and second levels, which in fact is a RAC drawing for $K_8$.  A zoom-in of the pdf file gives more details of the drawing.}
\label{fig:RACdrawing}
\end{figure}
\paragraph{\bf First bend}
Let $(x,y)$ be the coordinates of $V_{i,j}$, \emph{i.e.}, the $j$th vertex of the $i$th level.
For any $V_{i,j}$, $1\leq i \leq \sqrt{n}$, the first bends of the edges incident to $V_{i,j}$ are placed at $(x+(i-1)\sqrt{n}+1,~y+1), \dots, (x+i\sqrt{n}-j,~y+1), (x+i\sqrt{n}+1,~y+1), \dots, (x+n,~y+1)$. 


We denote by $a_k(V_{i,j})$, $k=(i-1)\sqrt{n}+1, \dots, i\sqrt{n}-j, i\sqrt{n}+1, \dots, n$ the first bends of the edges incident to $V_{i,j}$, from left to right respectively. Figure~\ref{fig:1stBend} depicts the first bends and the segments of the edges incident to $V_{1,1}$, $V_{2,2}$, $V_{3,3}$, and $V_{4,4}$ for $n=16$. 
Denote by $S_1$ the set of all the first segments of the edges incident to all the vertices. 
\begin{figure}
\centering
\includegraphics[scale=.8]{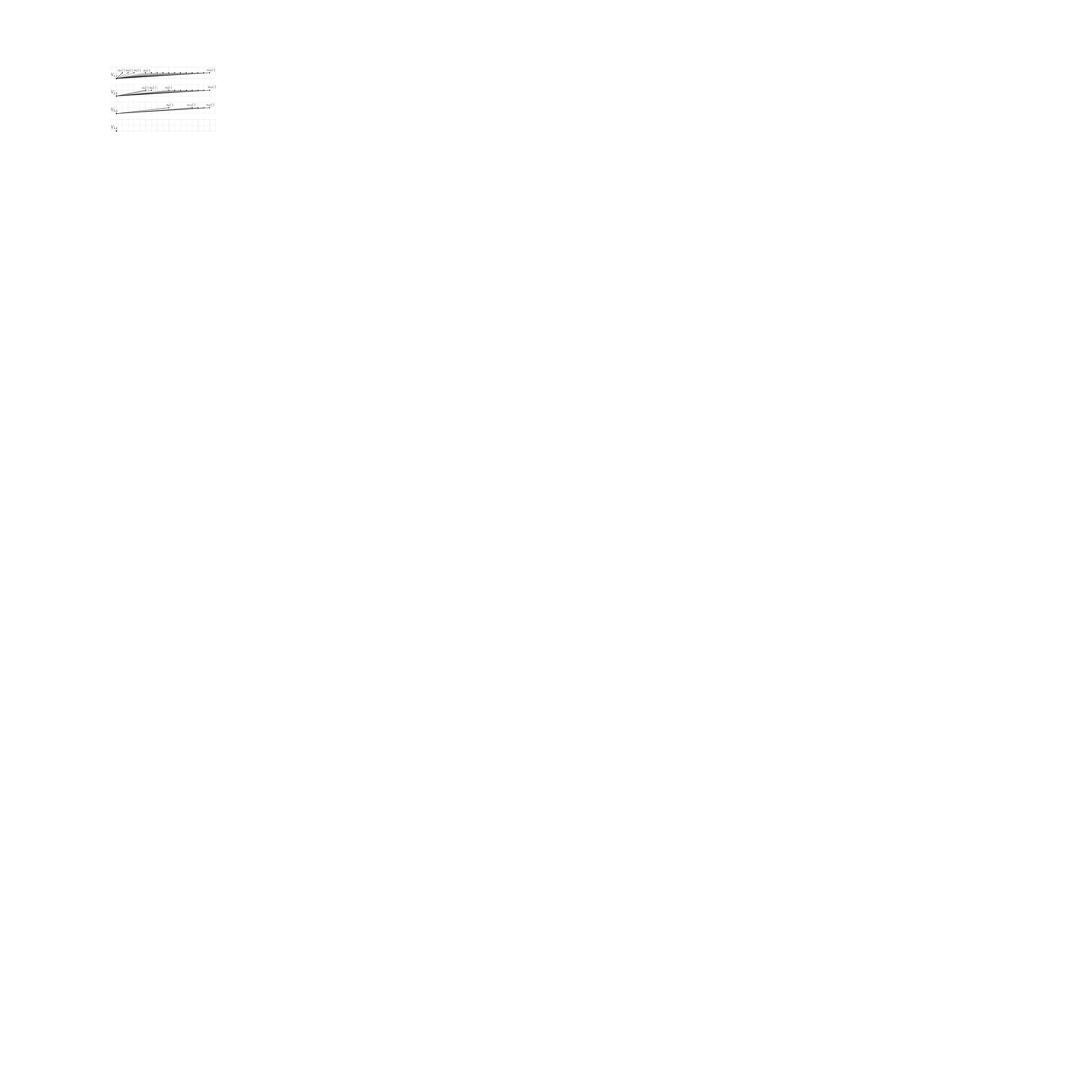}
\caption{The first bends and the segments of the edges incident to the vertex $V_{1,1}$, $V_{2,2}$, $V_{3,3}$, and $V_{4,4}$,  for $n=16$.}
\label{fig:1stBend}
\end{figure}
\paragraph{\bf Second bend}
Draw a segment with slope $\alpha$ from $a_k(V_{i,j})$ to some corresponding point, say $b_k(V_{i,j})$, such that the vertical distance between $a_k(V_{i,j})$ and $b_k(V_{i,j})$ is $(\sqrt{n}-j+i)\sqrt[4]{n}+1$. Denote by $S_2$ the set of all the second segments of the edges incident to all the vertices. For $n=16$, Figure~\ref{fig:2ndBend} depicts the second bends and the segments of the edges incident to $V_{1,4}$.
\begin{figure}
\centering
\includegraphics[scale=.5]{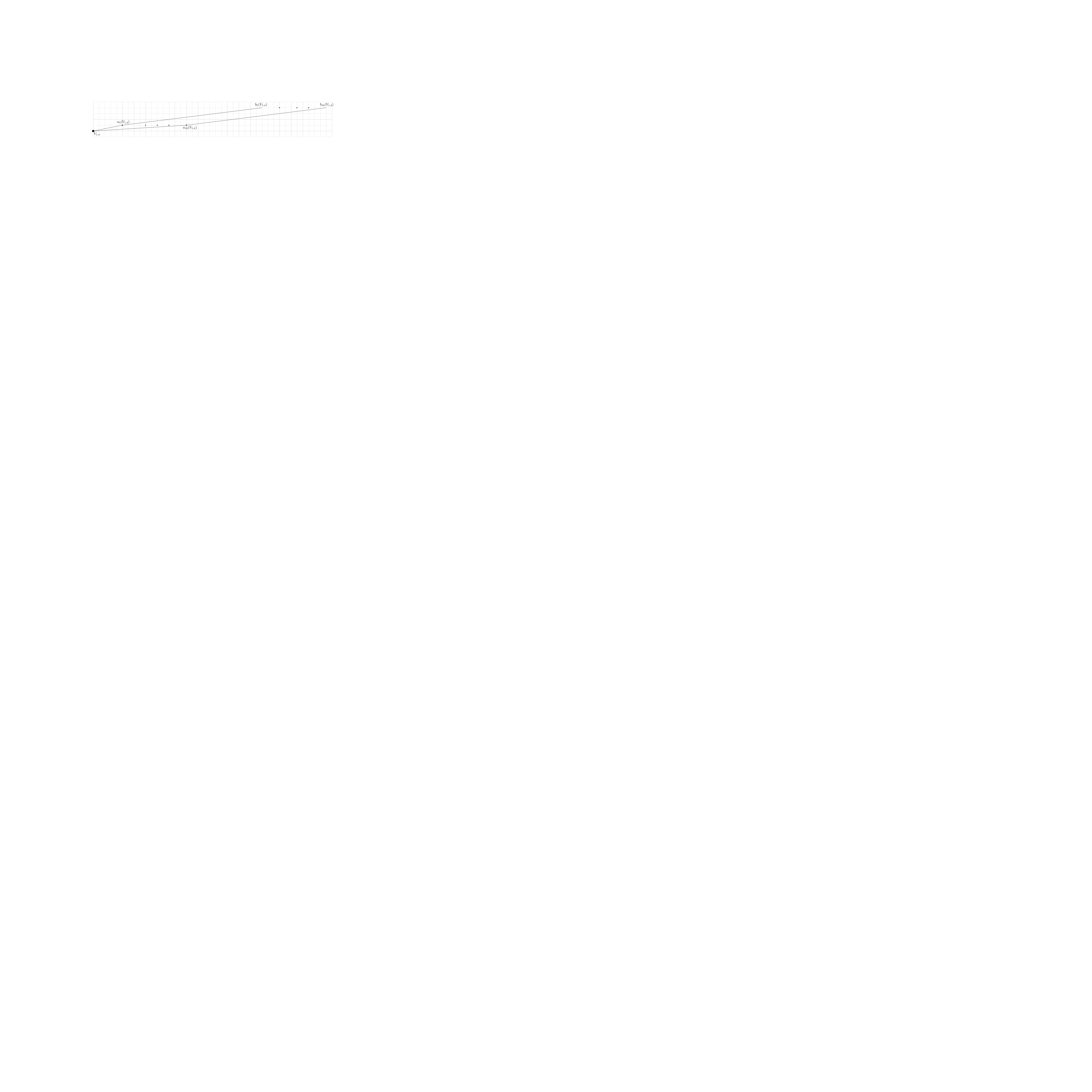}
\caption{The second bends, denoted by $b_1(V_{1,4}),\dots,b_{16}(V_{1,4})$, and the segments of the edges incident to $V_{1,4}$.}
\label{fig:2ndBend}
\end{figure}
\begin{figure}
\centering
\includegraphics[scale=.6]{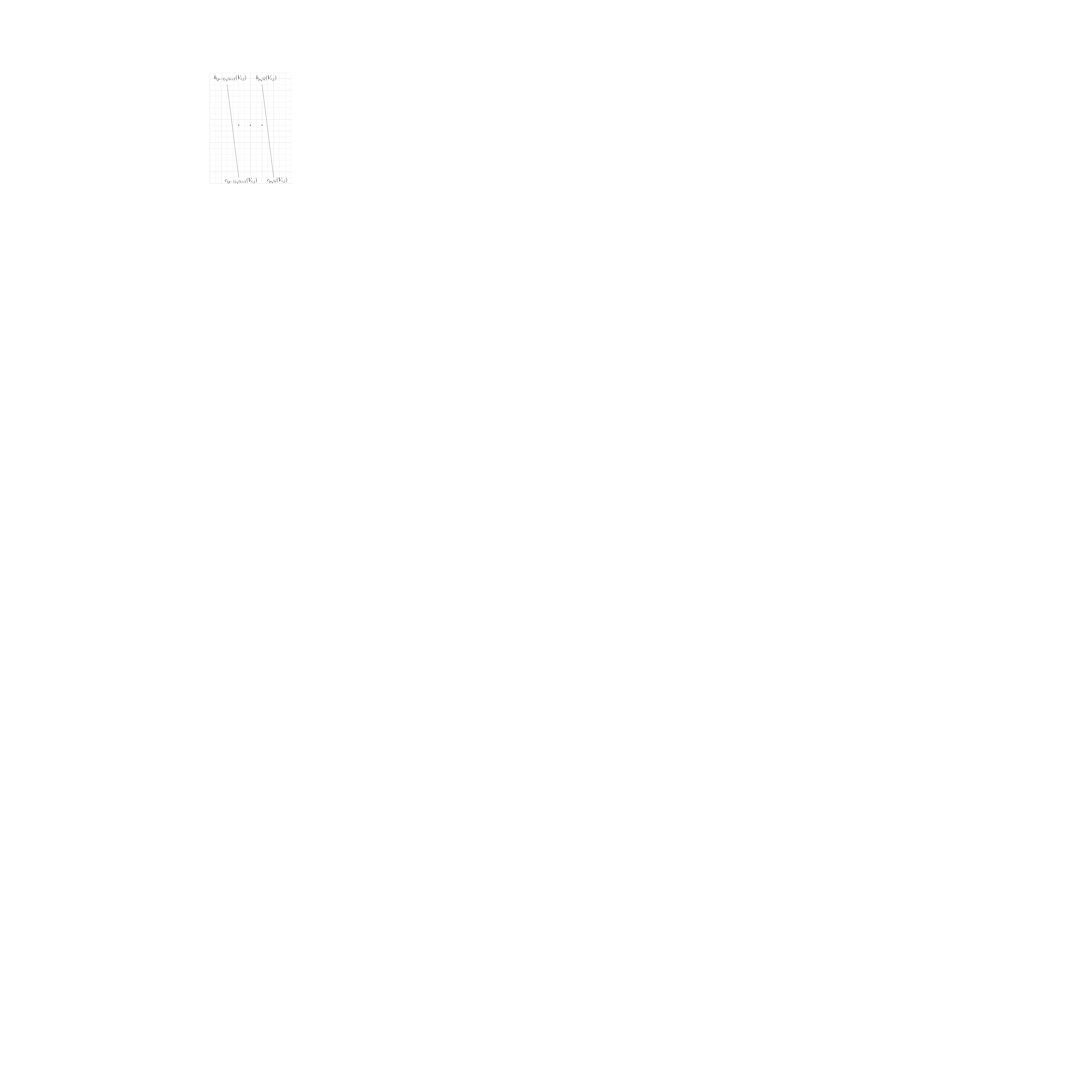}
\caption{The third bends, denoted by $c_{(p-1)\sqrt{n}+1}(V_{i,j}),\dots,c_{p\sqrt{n}}(V_{i,j})$), and the segments of the edges incident to some vertex $V_{i,j}$ at a level $i$, where $i<p$.}
\label{fig:3rdBend}
\end{figure}
\paragraph{\bf Third bend}
In our drawing technique, we draw the edges from a vertex at a {\it higher} level to some other vertex at a {\it lower} level. We used an idea for indices $k$ of the bends to easily find out where we want to bend in some lower level.
Let $p=\lceil{ k / \sqrt{n}} \rceil$.  
The value $p$ gives the label of a level where we want to make the third bend in order to connect to the $p$th level's vertices from higher level's vertices. 
Note that all the edges incident to vertices of {\it higher levels} whose second bends are indexed by $b_{(p-1)\sqrt{n}+1}(\cdot), \dots, b_{p\sqrt{n}}(\cdot)$  bend to get connected to the vertices at level $p$. 
Also note that, the edges incident to the vertices of {\it level $p$} whose second bends are indexed by  $b_{(p-1)\sqrt{n}+1}(\cdot), \dots, b_{p\sqrt{n}-1}(\cdot)$ bend to get connected to the vertices of the same level $p$.
%
In particular, 
we draw the third segment with slope $-1/\alpha$ from $b_k(V_{i,j})$ to some corresponding point, say $c_k(V_{i,j})$, such that the vertical distance between $b_k(V_{i,j})$ and $c_k(V_{i,j})$ is $(8p-8i+7)\sqrt[4]{n^3}$. Denote by $S_3$ the set of all the third segments of the edges incident to all the vertices. Figure~\ref{fig:3rdBend} depicts the third bends and the segments of the edges incident to some vertex $V_{i,j}$, where $i<p$.
\paragraph{\bf Fourth bend}
We draw the fourth segment with slope $\alpha$ from $c_k(V_{i,j})$ to enter to level $p$, where $p=\lceil{ k/ \sqrt{n}} \rceil$. Let $q=(k-1\mod\sqrt{n})$. For the points $c_k(V_{i,j})$ of the vertices of the vertices $V_{i,j}$, we draw the fourth segment  with slope $\alpha$ from $c_k(V_{i,j})$ to some corresponding point, say $d_k(V_{i,j})$, such that horizontal distance between $c_k(V_{i,j})$ and $d_k(V_{i,j})$ is $(i+q)n+\sqrt[4]{n^3}$. Denote by $S_4$ the set of all the fourth segments of the edges incident to all the vertices. 
Figure~\ref{fig:4thBend} depicts the fourth bends and the segments of the edges incident to some vertex $V_{i,j}$ of a level $i$ entering to level $p$, where $i<p$. 
\begin{figure}
\centering
\includegraphics[scale=.5]{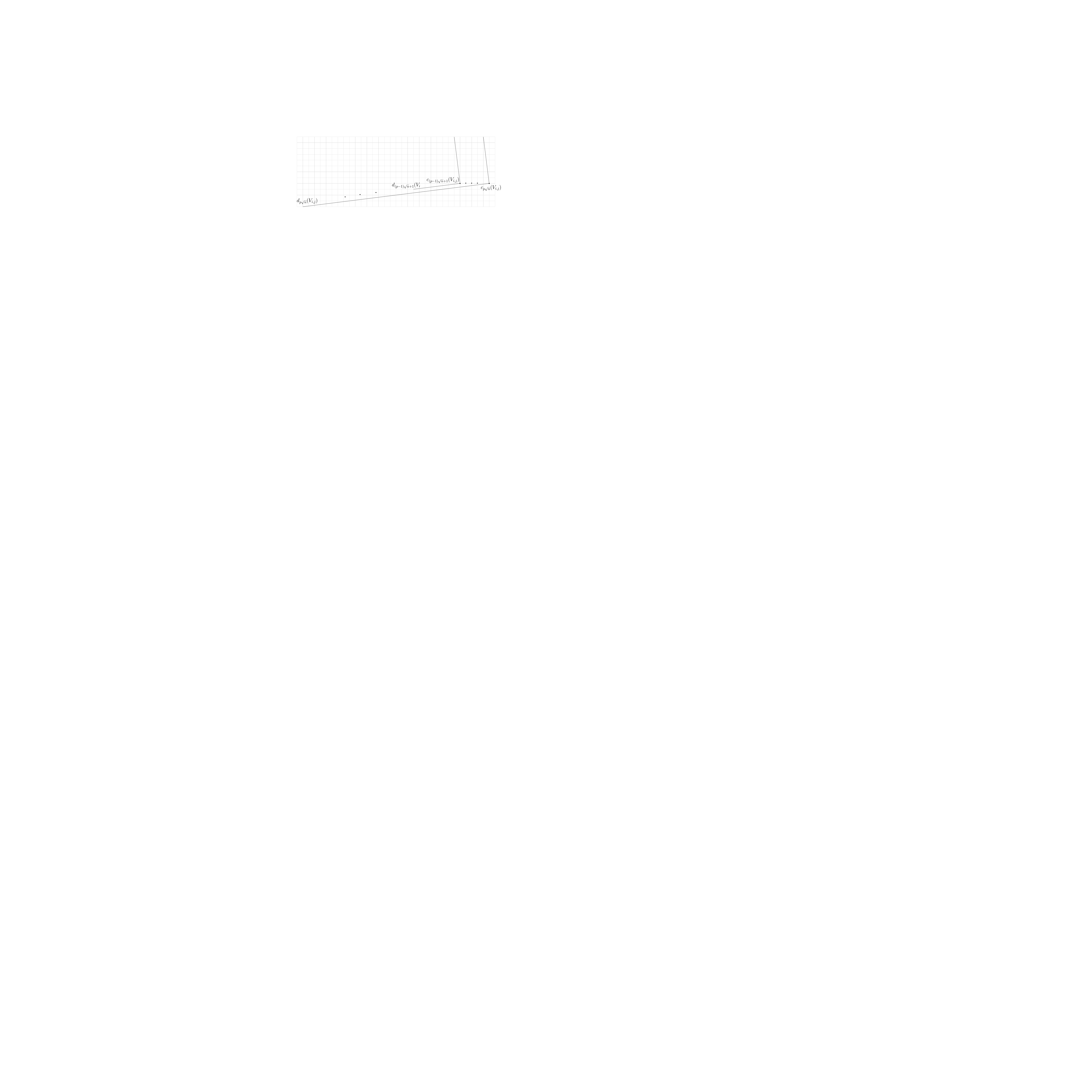}
\caption{The fourth bends, denoted by $d_{(p-1)\sqrt{n}+1}(V_{i,j}),\dots,d_{p\sqrt{n}}(V_{i,j})$, and the segments of the edges incident to some vertex $V_{i,j}$ at higher level $i$ entering to level $p$.}
\label{fig:4thBend}
\end{figure}
\paragraph{\bf Fifth and Sixth bends}
We draw the fifth segment with slope $-1/\alpha$ from $d_k(V_{i,j})$ to some corresponding point, say $e_k(V_{i,j})$, such that the vertical distance between $d_k(V_{i,j})$ and $e_k(V_{i,j})$ is $3\sqrt[4]{n^3}$. Denote by $S_5$ the set of all the fifth segments of the edges incident to all the vertices.
Figure~\ref{fig:4thBend2} depicts the fifth bends and the segments of all the edges incident to vertices of the highest level, for $n=16$.

Finally, we draw the sixth segment from $e_k(V_{i,j})$ to some corresponding point, say $f_k(V_{i,j})$, which has the same horizontal coordinate as of $e_k(V_{i,j})$ and the vertical coordinate of $f_k(V_{i,j})$ is the vertical coordinate of $V_{p,j}$ minus one.
Denote by $S_6$ the set of all the sixth segments of the edges incident to all the vertices.

Note that the last bends $f_k(V_{i,j})$ of the edges incident to vertices will be connected to $V_{p,\sqrt{n}-q}$.
Figure~\ref{fig:56thBend} depicts the sixth bends and the segments (and seventh segments) of all the edges incident to a vertex of some higher level $i$ entering to level $p$, where $i<p$.
Denote by $S_7$ the set of all the seventh segments of the edges incident to all the vertices.
\begin{figure}
\centering
\includegraphics[scale=.55]{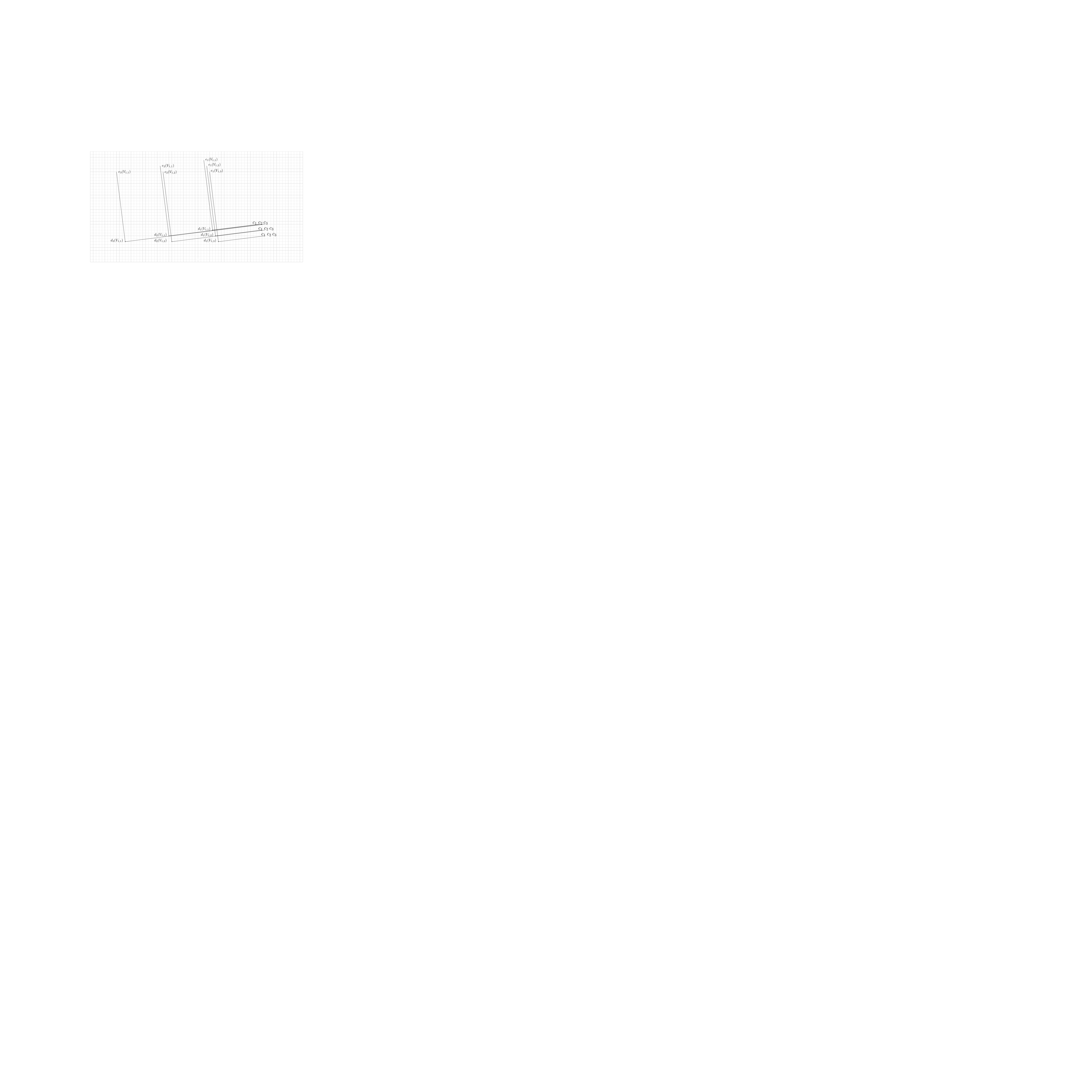}
\caption{The fifth bends and the segments of the edges incident to the vertices $V_{1,1},V_{1,2},V_{1,3}$, and $V_{1,4}$ of the highest level entering to the same level, for $n=16$. }
\label{fig:4thBend2}
\end{figure}

\begin{figure}
\centering
\includegraphics[scale=.5]{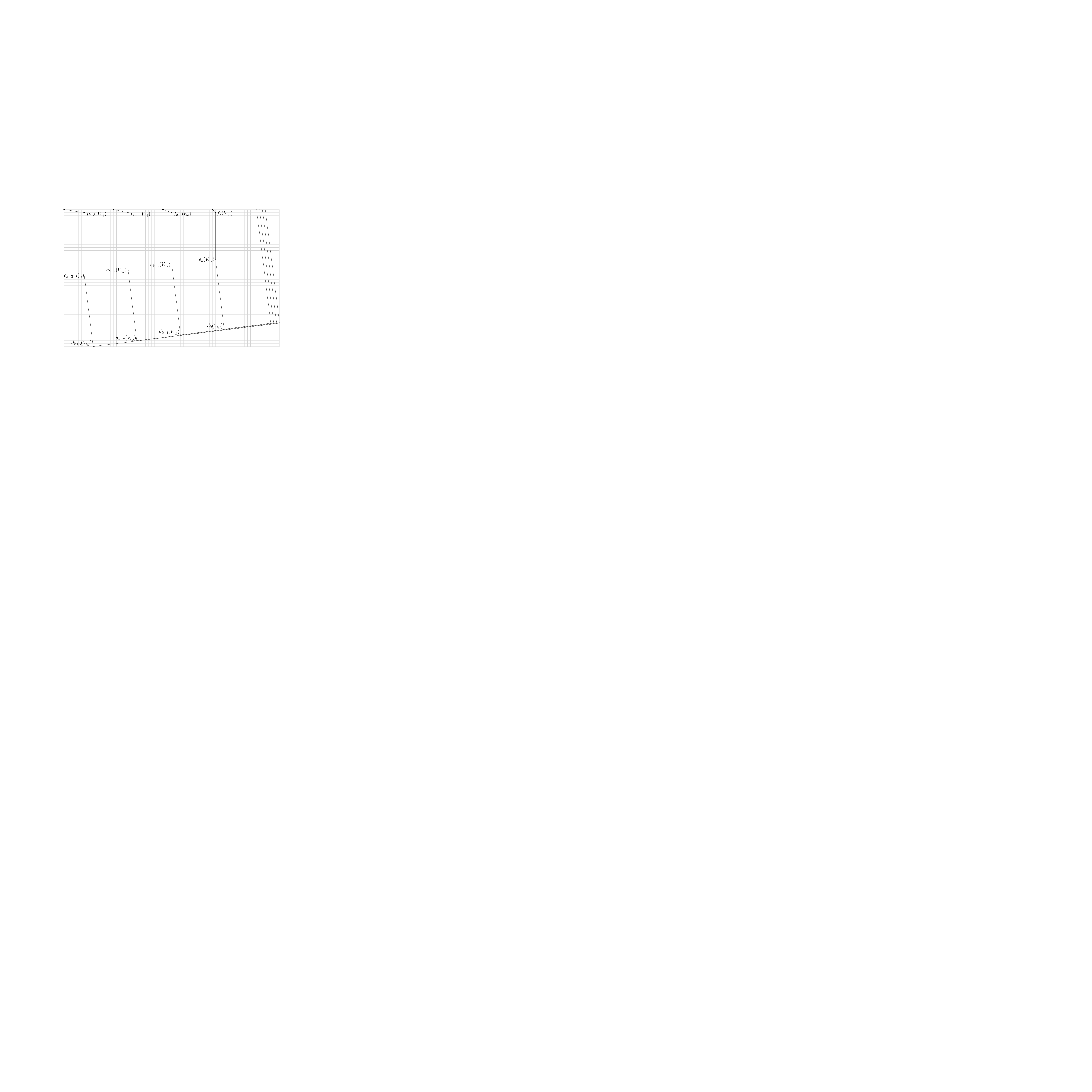}
\caption{The sixth bends of the edges incident to some vertex $V_{i,j}$ of a higher level $i$ entering to some level $p$, where $i<p$. }
\label{fig:56thBend}
\end{figure}

Now we can summarize our technique for drawing the edges of $K_n$. For each pair of vertices $V_{i,j}=(x(V_{i,j}),y(V_{i,j}))$, the $j$th vertex at level $i$ and $V_{u,w}=(x(V_{u,w}),y(V_{u,w}))$, the $w$th vertex at level $u$, where $i\leq u$,  the six bends of the edge connecting from $V_{i,j}$ to $V_{u,w}$ are placed at $a, b, c, d, e$ and $f$:
\begin{itemize}
\item $a =(x(V_{i,j})+(u-1)\sqrt{n}+\sqrt{n}-w+1,~y(V_{i,j})+1)$,
\item $b =(x(a)+(\sqrt{n}-j+i)n+\sqrt[4]{n^3},~y(a)+(\sqrt{n}-j+i)\sqrt[4]{n})+1)$,
\item $c =(x(b)+8(u-i+1)-1,~y(b)-(8(u-i+1)-1)\sqrt[4]{n^3})$,
\item $d =(x(c)-(i+\sqrt{n}-w)n-\sqrt[4]{n^3},~y(c)-(i+\sqrt{n}-w)\sqrt[4]{n}-1)$,
\item $e =(x(d)-3,~y(d)+3\sqrt[4]{n^3})$,
\item $f =(x(e),~y(V_{u,w})-1)$.
\end{itemize}
Since there are $n$ placements for the vertices of $K_n$ in $\sqrt{n}$ levels, and each edge has six bends, the following lemma holds.
\begin{lemma}\label{the:ConsTime}
Our drawing algorithm for $K_n$ takes time $O(n+m)=O(n^2)$, where $m$ is the number of edges in $K_n$. 
\end{lemma}

\begin{lemma}\label{the:KnResult}
The complete graph $K_n$ admits a RAC drawing with curve complexity $6$ and area $O(n^{11/4})$. The drawing  can be computed in $O(n+m)$ time.
\end{lemma}
\begin{proof}
The segments in $S_r$, for any $1\leq r \leq 7$, do not cross each other except possibly at their endpoints. The crossings can only happen in three cases: Between a segment in $S_2$ and a segment in $S_3$, between a segment in $S_3$ and a segment in $S_4$, and between a segment in  $S_4$ and a segment in $S_5$. In all these cases, the two crossing segments are orthogonal. Thus our approach gives a RAC drawing, which can be computed in $O(n+m)$ time (by Lemma~\ref{the:ConsTime}).

Since the \textit{leftmost} point of the drawing is $V_{1,1}$ and the \textit{rightmost} point is $c_{n-1}(V_{\sqrt{n},1})$, the {\it width} of the grid is $x(c_{n-1}(V_{\sqrt{n},1})) - x(V_{1,1}) = 2n\sqrt{n}+n+\sqrt[4]{n^3}+7\sqrt{n}-8 = O(n\sqrt{n})$.
Also, since the \textit{topmost} point is $b_{k}(V_{1,1})$, $1\leq k \leq n-1$, and the \textit{bottommost} point is $d_{n}(V_{\sqrt{n}-1,\sqrt{n}})$, the {\it height} of the grid is $y(b_{k}(V_{1,1})) - y(d_{n}(V_{\sqrt{n}-1,\sqrt{n}}))= 8n\sqrt[4]{n}+2\sqrt[4]{n^3}+\sqrt{n}-3\sqrt[4]{n}-1 = O(n\sqrt[4]{n})$.
Therefore, the grid size is $O(n^{11/4})$.

For the case when $\sqrt[4]{n}$ is not an integer, our technique still works and has the same complexities: Let $l= {\lceil \sqrt[4]{n} \rceil}$. For this case, we consider $l^2$ levels, each with at most $l^2$ vertices, and in the calculations of our technique we replace $n$ by $l^4$, $\sqrt{n}$ by $l^2$, and $\sqrt[4]{n}$ by $l$. Thus the grid size will be $O(l^{11})= O({\lceil \sqrt[4]{n} \rceil}^{11})=O(n^{11/4})$.
\end{proof}

From the fact that every graph with $n$ vertices is a subgraph of $K_n$, our results in Lemma~\ref{the:KnResult} holds for any given graph.
The following states the main result of this paper:
\begin{theorem}\label{the:main}
Every graph with $n$ vertices and $m$ edges admits a RAC drawing with curve complexity $6$ and area $O(n^{11/4})$, and the drawing can be computed in time $O(n+m)$.
\end{theorem}




\begin{thebibliography}{10}
\expandafter\ifx\csname url\endcsname\relax
  \def\url#1{\texttt{#1}}\fi
\expandafter\ifx\csname urlprefix\endcsname\relax\def\urlprefix{URL }\fi
\expandafter\ifx\csname href\endcsname\relax
  \def\href#1#2{#2} \def\path#1{#1}\fi

\bibitem{DidimoSurvey2019}
W.~Didimo, G.~Liotta, F.~Montecchiani, A survey on graph drawing beyond
  planarity, ACM Comput. Surv.~(1) (2019) 4:1--4:37.

\bibitem{Purchase1997}
H.~Purchase, Which aesthetic has the greatest effect on human understanding?,
  in: G.~Di~Battista (Ed.), Graph Drawing, Springer Berlin Heidelberg, 1997,
  pp. 248--261.

\bibitem{DIDIMO20115156}
W.~Didimo, P.~Eades, G.~Liotta, Drawing graphs with right angle crossings,
  Theoretical Computer Science 412~(39) (2011) 5156--5166.

\bibitem{4126225}
W.~{Huang}, Using eye tracking to investigate graph layout effects, in: 2007
  6th International Asia-Pacific Symposium on Visualization, 2007, pp. 97--100.

\bibitem{4475457}
W.~{Huang}, S.~{Hong}, P.~{Eades}, Effects of crossing angles, in: 2008 IEEE
  Pacific Visualization Symposium, 2008, pp. 41--46.

\bibitem{HUANG2014452}
W.~Huang, P.~Eades, S.-H. Hong, Larger crossing angles make graphs easier to
  read, Journal of Visual Languages \& Computing 25~(4) (2014) 452--465.

\bibitem{SOFSEM2011Argyriou}
E.~N. Argyriou, M.~A. Bekos, A.~Symvonis, The straight-line {RAC} drawing
  problem is {NP}-hard, in: I.~{\v{C}}ern{\'a}, T.~Gyim{\'o}thy,
  J.~Hromkovi{\v{c}}, K.~Jefferey, R.~Kr{\'a}lovi{\'{c}}, M.~Vukoli{\'{c}},
  S.~Wolf (Eds.), SOFSEM 2011: Theory and Practice of Computer Science,
  Springer Berlin Heidelberg, 2011, pp. 74--85.

\bibitem{BEKOS201748}
M.~A. Bekos, W.~Didimo, G.~Liotta, S.~Mehrabi, F.~Montecchiani, On {RAC}
  drawings of 1-planar graphs, Theoretical Computer Science 689 (2017) 48--57.

\bibitem{GD201Angelini8}
P.~Angelini, M.~A. Bekos, H.~F{\"o}rster, M.~Kaufmann, On {RAC} drawings of
  graphs with one bend per edge, in: T.~Biedl, A.~Kerren (Eds.), Graph Drawing
  and Network Visualization, Springer International Publishing, 2018, pp.
  123--136.

\bibitem{ARIKUSHI2012169}
K.~Arikushi, R.~Fulek, B.~Keszegh, F.~Mori\'c, C.~D. T\'oth, Graphs that admit
  right angle crossing drawings, Computational Geometry 45~(4) (2012) 169--177.

\bibitem{DiGiacomo2011}
E.~Di~Giacomo, W.~Didimo, G.~Liotta, H.~Meijer, Area, curve complexity, and
  crossing resolution of non-planar graph drawings, Theory of Computing Systems
  49~(3) (2011) 565--575.

\bibitem{Pach2013}
W.~Didimo, G.~Liotta, The Crossing-Angle Resolution in Graph Drawing, Springer
  New York, New York, NY, 2013, pp. 167--184.

\bibitem{GD2019Foerster}
H.~F{\"o}rster, M.~Kaufmann, On compact {RAC} drawings, in: D.~Archambault,
  C.~D. Toth (Eds.), Graph Drawing and Network Visualization, Springer
  International Publishing, 2019, pp. 599--601.

\end{thebibliography}

\end{document}